\RequirePackage{fix-cm}
\documentclass[smallextended]{svjour3}       
\smartqed  
\usepackage{graphicx}

\usepackage{amsfonts}
\usepackage{amsmath}
\usepackage{amssymb}
\usepackage{latexsym}
\newcommand{\F}{\mathbb{F}}
\newcommand{\Aut}{\mbox{\rm Aut}}
\newcommand{\wt}{\mbox{\rm wt}}
\def\ds{\displaystyle}
\journalname{Designs, Codes and Cryptography}
\begin{document}

\title{On the Automorphisms of Order 15 for a Binary Self-Dual $[96, 48, 20]$ Code
}


\author{Stefka Bouyuklieva         \and
        Wolfgang Willems \and Nikolay Yankov
}


\institute{S. Bouyuklieva \at
              Faculty of Mathematics and Informatics,
Veliko Tarnovo University, Bulgaria,\\
                            \email{stefka@uni-vt.bg}           
           \and
           W. Willems \at
             Otto-von-Guericke Universit\"at, Magdeburg, Germany,\\
             \email{willems@ovgu.de}
             \and N. Yankov \at
             Faculty of Mathematics and Informatics,
 Shumen University, Bulgaria,\\
 \email{jankov\_niki@yahoo.com}
}

\date{Received: date / Accepted: date}

\maketitle

\begin{abstract}
The structure of binary self-dual codes invariant under the
action of a cyclic group of order $pq$ for odd primes $p\neq q$ is
considered. As an application we prove the nonexistence of an extremal self-dual $[96, 48,
20]$ code  with an automorphism of order $15$ which closes a gap in \cite{JW}.
\end{abstract}

\section{Introduction}
\label{intro}
Let $C=C^\perp$ be a binary self-dual code of length $n$ and
minimum distance $d$. A binary code is doubly-even if the weight
of every codeword is divisible by four. Self-dual doubly-even
codes exist only if $n$ is a multiple of eight. Rains \cite{Rains}
proved that the minimum distance $d$ of a binary self-dual
$[n,k,d]$ code satisfies the following bound:
$$d\le 4 \lfloor n/24 \rfloor +4, \ \ \ {\rm if} \ n\not\equiv 22\pmod {24},$$
$$d\le 4 \lfloor n/24 \rfloor +6, \ \ \ {\rm if} \ n\equiv 22\pmod {24}.$$
Codes achieving this bound are called extremal. If $n$ is a
multiple of 24, then a self-dual code meeting the bound must be
doubly-even \cite{Rains}. Moreover, for any nonzero weight $w$ in
such a code, the codewords of weight $w$ form a 5-design
\cite{Assmus}. This is one reason  why
 extremal codes of length $24m$ are of particular interest. Unfortunately, only for $m=1$ and $m=2$
 such codes are known, namely the $[24, 12, 8]$ extended Golay  code and the $[48, 24, 12]$  extended quadratic residue code
 (see \cite{Handbook-SelfDual}). To date, the existence of no other extremal code of length $24m$ is known.
For $n=96$, only the primes $2,3$ and $5$ may divide the order of the automorphism group of the extremal code 
and the cycle structure of prime  order automorphisms are
as follows
\begin{equation}\label{primes} {
\begin{tabular}{c|c|c}
$p$ & \mbox{\rm number of $p$-cycles} & \mbox{ \rm number of fixed points} \\
\hline
$2$ & $48$ & $0$ \\
$3$ & $ 30,32$ & $6,0 $\\
$5$ & $18$ & $6$ \\
\end{tabular} }
\end{equation}
(see  Theorem,
part a) in \cite{JW}). We would like to mention here that in part
b) of the Theorem (the case where elements of order $3$ are acting fixed
point freely)  four orders of possible automorphism groups are missing, namely  $15, 30,
240$ and $480$. The gap is due to the  fact that the existence of
elements of order $15$ with six cycles of length $15$ and two
cycles of length $3$  are not excluded in the given proof. We close
this gap by proving

\begin{theorem}\label{Main_theorem}  A binary doubly-even $[96,48,20]$
self-dual code with an automorphism of order 15 does not exist.
\end{theorem}

This note consists of three sections. Section 2 is devoted to some
theoretical results on  binary self-dual codes invariant under
the action of a cyclic group. In Section 3 we study the structure of a
putative extremal self-dual $[96, 48, 20]$ code having an automorphism of order 15.
Using this structure and combining the possible subcodes we prove Theorem \ref{Main_theorem}.
In an additional section, namely Section 4, we prove that an extremal self-dual code of length 96 does not have automorphisms of type 3-(28,12).
This assertion is used by other authors but no proof has been published so far.

\section{Theoretical results}
\label{sec:1}

Let $C$ be a binary linear code of length $n$ and let $\sigma$ be an
automorphism of $C$ of order $r$ where $r$ is odd (not necessarily a prime). Let
\begin{equation}\label{general sigma}
\sigma =\Omega_1\Omega_2 \dots\Omega_m
\end{equation}
be the factorization of $\sigma$ into disjoint cycles (including
the cycles of length 1). If $l_i$ is the length of the cycle
$\Omega_i$ then ${\rm lcm}(l_1,\dots,l_m)=r$ and $l_i$ divides $r$.
Therefore $l_i$ is odd for $i=1,\dots,m$ and $1\leq l_i\leq r$.

Let $F_{\sigma} (C)=\{ v\in C:v\sigma=v\}$ and
$$E_{\sigma}(C)=\{ v\in C:wt(v\vert \Omega_i)\equiv 0 \ (\bmod \ 2), \, i=1,\dots,m\},$$
 where $v\vert\Omega_i$ is the restriction of $v$ on $\Omega_i$.
 The following theorem is similar to Theorem 1 from \cite{Huff86} but Huffman's result is due to an automorphism that has only $c$ $r$-cycles and $f$ fixed points. We consider automorphisms of odd order $r$ which can be factorized into independent cycles of different lengths.

 \begin{theorem}
 The code $C$ is a direct sum of the subcodes $F_{\sigma}(C)$ and $E_{\sigma}(C)$.
\end{theorem}

\begin{proof} We follow the proof of Lemma 2 in
\cite{Huff}. Obviously, $F_{\sigma}(C)\cap
E_{\sigma}(C)= \{0\}$. Let $v\in C$ and
$w=v+\sigma(v)+\cdots+\sigma^{r-1}(v)$. Since $w\in C$ and
$\sigma(w)=w$ we get  $w\in F_{\sigma}(C)$.

On the other hand,
$\wt(\sigma^j(v)\vert_{\Omega_i})=\wt(v\vert_{\Omega_i})$ for all
$i=1,2,\dots,m$ and $j\ge 1$. Hence
$\sigma(v)+\cdots+\sigma^{r-1}(v)\vert_{\Omega_i}$ is a sum of an
even number of vectors of the same weight. Thus
$\wt(\sigma(v)+\cdots+\sigma^{r-1}(v)\vert_{\Omega_i})$ is even for
$i=1,2,\dots,m$. It follows that
$u=\sigma(v)+\cdots+\sigma^{r-1}(v)\in E_{\sigma}(C)$. So
$v=w+u\in F_{\sigma}(C)+E_{\sigma}(C)$ which proves that
$C=F_{\sigma}(C)\oplus E_{\sigma}(C)$.
\end{proof}

Let $\F_2^n$ be the $n$-dimensional vector space over the binary field $\F_2$, and let $\pi :F_{\sigma}(C)\to\mathbb{F}_{2}^{m} $ be the projection
    map, i.e.,  $ (\pi(v))_i =v_j $ for some
    $ j\in\Omega_i$ and $ i=1,2,\ldots,m$.
Clearly, $v\in F_{\sigma}(C)$ iff $v\in C$ and $v$ is constant on
    each cycle.  The following theorem is similar to Theorem 3 from \cite{Huff86} but as for the previous theorem, Huffman's result is due to an automorphism that has only $c$ $r$-cycles and $f$ fixed points. We consider here binary codes having an automorphism $\sigma$ of odd order $r$ without other restrictions.

\begin{theorem}
 If $C$ is a binary self-dual code with an
automorphism $\sigma$ of 
odd order  then
$C_{\pi}=\pi(F_{\sigma}(C))$ is a binary self-dual code of length
$m$. \end{theorem}

\begin{proof} Let $v,w\in F_{\sigma}(C)$. If $\langle \cdot  \, ,\cdot \rangle$ denotes the Euclidean inner product on $\F_2^n$ then $\langle v,w\rangle=\langle\pi(v),\pi(w)\rangle=0$ since $l_i$ is odd for all $i$.
 Hence $C_{\pi}$ is a self-orthogonal code.
If $u\in C_{\pi}^\perp$ and $u'=\pi^{-1}(u)$ then $\langle
u',v\rangle=\langle u,\pi(v)\rangle=0$ for all $v\in F_\sigma(C)$.
Furthermore,
 $\langle u',v\rangle = \sum_{i=1}^m \langle
u'\vert_{\Omega_i},v\vert_{\Omega_i}\rangle= 0$
for all $v\in E_\sigma(C)$ since $u'$ is constant on $\Omega_i$ and $\wt(v|_{\Omega_i})$ is even. Thus $u'\in C^\perp=C$. Hence, $u' \in F_\sigma(C)$ and
 therefore $u=\pi(u')\in C_{\pi}$ which proves that $C_{\pi}$ is a
self-dual code.
\end{proof}

\begin{theorem}{\rm \cite{Huff86}}
Let $C$ be a binary self-dual code of length $n=cr+f$ and let $\sigma$
be an automorphism of $C$ of odd order $r$ such that
\begin{equation}\label{cr+f}
\sigma =\Omega_1\dots\Omega_c\Omega_{c+1}\dots\Omega_{c+f}
\end{equation}
where $\Omega_i=((i-1)r+1,\dots,ir)$ are cycles of
length $r$ for $ i=1,\dots,c$, and $\Omega_{c+i}=(cr+i)$ are the fixed
points for $i=1,\dots,f$. Then $F_{\sigma}(C)$ and $E_{\sigma}(C)$ have dimension
$(c+f)/2$ and $c(r-1)/2$, respectively.
\end{theorem}


If $\sigma$ is of prime order $p$ with $c$ cycles of length $p$ and $f$ fixed points we say that $\sigma$ is of type
$p$-$(c,f)$.

\subsection{Connections with quasi-cyclic codes}

For further investigations, we need two theorems concerning
the theory of finite fields and cyclic codes. Let $r$ be a
positive integer coprime to the characteristic of the
field $\mathbb{F}_l$ of cardinality $l$, where $l$ is the power of a prime. Consider the factor ring
$\mathcal{R}=\mathbb{F}_l[x]/(x^r-1)$, where $(x^r-1)$ is the
principal ideal in $\mathbb{F}_l[x]$ generated by $x^r-1$. Let
$$x^r-1=f_0(x)f_1(x)\dots f_s(x)$$
be the factorization of $x^r-1$ into irreducible factors $f_i(x)$ over
$\F_l$ where $f_0(x)=x-1$. Let
 $I_j=\langle
\ds\frac{x^r-1}{f_j(x)}\rangle$ be the ideal of $\mathcal{R}$
generated by $\ds\frac{x^r-1}{f_j(x)}$ for $j=0,1,\dots,s$.  Finally, by
$e_j(x)$ we denote the generator idempotent of $I_j$;
 i.e., $e_j(x)$ is the identity of the two-sided ideal $I_j$.
With these notations we have the following well-known result.

\begin{theorem}
{\rm (see \cite{pless-libro})}

 (i) $\mathcal{R}=I_0\oplus
I_1\oplus\cdots\oplus I_s$.

 (ii) $I_j$ is a field which is isomorphic to the field
$\mathbb{F}_{l^{deg(f_j(x))}}$ for $j=0,1,\dots,s$.

 (iii) $e_i(x)e_j(x)=0$ for $i\neq j$.

 (iv) $\sum_{j=0}^s e_j(x)=1$.
\end{theorem}


According to \cite{LSoleI}, there is a decomposition $$x^r-1=g_0(x)g_1(x)\cdots
g_m(x)h_1(x)h_1^*(x)\cdots h_t(x)h_t^*(x),$$ where $s=m+2t$ and
$\{g_0,g_1,\ldots g_m, h_1,h_1^*,\ldots, h_t,h_t^*\}=\{
f_0,f_1,\ldots,f_s\}$.
Furthermore, $h_i^*(x)$ is the reciprocal polynomial of
$h_i(x)$, $h_i^*\neq h_i$ for $i=1,\ldots,t$ and $g_i(x)$ coincides with
its reciprocal polynomial where $g_0(x)=f_0(x)=x-1$. Finally,  we denote the
field $\langle \frac{x^r-1}{g_j(x)}\rangle$
by $G_j$ for $j=0,1,\ldots,m$, $\langle \frac{x^r-1}{h_j(x)}\rangle$
by $H_j$ for $j=1,\ldots,t$, and $\langle \frac{x^r-1}{h^*_j(x)}\rangle$  by $H^*_j$ for
$j=1,\ldots,t$.


To continue the investigations, we need to prove some properties
of  binary linear codes of length $cr$ with an automorphism
$\tau$ of order $r$ which has $c$ independent $r$-cycles. If $C$ is
such a code then $C$ is a quasi-cyclic code of length $cr$ and
index $c$. Next, we define a map $\phi:\F_2^{cr}\to\mathcal{R}^c$ by
$$\phi(v)=(v_0(x),v_1(x),\dots,v_{c-1}(x))\in \mathcal{R}^c, $$
where $ v_i(x)=\sum_{j=0}^{r-1}v_{ij}x^j$ and $
(v_{i0},\dots,v_{i,r-1})=v|_{\Omega_i}.$
Clearly, $\phi(C)$ is a
linear code over the ring $\mathcal{R}$ of length $c$. Moreover, according to  \cite{LSoleI},
we have $\phi(C)^\perp=\phi(C^\perp)$ where the dual code $C^\perp$ over $\F_2$ is
taken under the Euclidean inner product, and the dual code  $\phi(C)^\perp$ in
$\mathcal{R}^c$ is taken with respect to the following Hermitian
inner product:
$$\langle
u,v\rangle=\sum_{i=0}^{c-1}u_i\overline{v}_i\in\mathcal{R}^c, \ \
\overline{v}_i=v_i(x^{-1})=v_i(x^{r-1}).$$ In particular, the
quasi-cyclic code $C$ is self-dual if and only if $\phi(C)$ is
self-dual over $\mathcal{R}$ with respect to the Hermitian inner
product.

Every linear code $C$ over the ring $\mathcal{R}$ of length $c$
can be decomposed as a direct sum
$$C=(\bigoplus_{i=0}^m C_i)\oplus(\bigoplus_{j=1}^t(C'_j\oplus
C''_j)),$$ where $C_i$ is a linear code over the field $G_i \ (
i=0,1,\dots,m)$, $C'_j$ is a linear code over $H_j$ and $C''_j$ is
a linear code over $H^*_j \  (j=1,\dots,t)$.

\begin{theorem} {\rm  (see \cite{LSoleI})}
A linear code $C$ over $\mathcal{R}$ of length $c$ is self-dual
with respect to the Hermitian inner product, or equivalently a
$c$-quasi-cyclic code of length $cr$ over $\F_q$ is self-dual with
respect to the Euclidean inner product, if and only if
$$C=(\bigoplus_{i=0}^m C_i)\oplus(\bigoplus_{j=1}^t(C'_j\oplus
(C'_j)^\perp)),$$ where $C_i$ is a self-dual code over $G_i$ for
$i=0,1,\dots,m$ of length $c$ (with respect to the Hermitian inner
product) and $C_j'$ is a linear code of
length $c$ over $H_j$ and $(C'_j)^\perp$ is its dual with respect
to the Euclidean inner product  for $1\leq j\leq t$,.
\end{theorem}

\subsection{The case $r=pq$}

We consider now the case  $r=pq$ for different odd primes $p$
and $q$ such that $2$ is a primitive root modulo $p$ and modulo
$q$. The ground field  is $\F_2$. Then
$$x^r-1=(x-1)Q_p(x)Q_q(x)Q_r(x)=(1+x)(1+x+\cdots+x^{p-1})(1+x+\cdots+x^{q-1})Q_r(x)$$
where $Q_i(x)$ is the $i$-th cyclotomic polynomial. Moreover, both
$Q_p(x)$ and $Q_q(x)$ are irreducible over $\F_2$ since $2$ is primitive modulo $p$ and modulo $q$ as well.
Finally, if
$$Q_r(x)=g_3(x)\dots g_m(x)h_1(x)h_1^*(x)\cdots h_t(x)h_t^*(x)$$ is
the factorization of the $r$-th cyclotomic polynomial into
irreducible factors over $\mathbb{F}_2$, then these factors have
the same degree, namely
$\frac{\phi(r)}{m-2+2t}=\frac{(p-1)(q-1)}{m-2+2t}$, where $\phi$
is  Euler's phi function. \\

Let $$\sigma
=\Omega_1\dots\Omega_c\Omega_{c+1}\dots\Omega_{c+t_q}\Omega_{c+t_q+1}\dots\Omega_{c+t_q+t_p}\Omega_{c+t_q+t_p+1}\dots\Omega_{c+t_q+t_p+f}
$$
where \\ $\Omega_i=((i-1)r+1,\dots,ir)$ are cycles of
length $pq$ for $ i=1,\dots,c$, \\
 $\Omega_{c+i}=(cr+(i-1)q+1,\dots,cr+iq)$
 are cycles of length $q$ for $i=1,\dots,t_q$, \\
$\Omega_{c+t_q+i}=(cr+t_qq+(i-1)p+1,\dots,cr+t_qq+ip)$
 are cycles of length $p$ for $i=1,\dots,t_p$, and
$\Omega_{c+t_q+t_p+i}=(c+t_q+t_p+i)$ are the fixed
points for $ i=1,\dots,f$.                       \\

Let $E_\sigma(C)^*$ be the shortened code of $E_\sigma(C)$
obtained by removing the last $t_qq+t_pp+f$ coordinates from
the codewords having 0's there. Let $C_\phi=\phi(E_\sigma(C)^*)$.
Since $E_\sigma(C)^*$ is a binary quasi-cyclic code of length $cr$
and index $c$, $C_\phi$ is a linear  code over the ring
$\mathcal{R}$ of length $c$. Moreover
$$C_\phi=(\bigoplus_{i=0}^m M_i)\oplus(\bigoplus_{j=1}^t(M'_j\oplus
M''_j)),$$ where $M_i$ is a linear code over the field $G_i$,
$i=1,\dots,m$, $M'_j$ is a linear code over $H_j$ and $M''_j$ is a
linear code over $H^*_j$, $j=1,\dots,t$. For the dimensions we
have
\begin{align*}
    \dim E_\sigma(C)^*= \dim C_\phi & = (p-1)\dim M_1+(q-1)\dim M_2 \\
     & +\frac{(p-1)(q-1)}{m-2+2t}(\sum_{i=3}^m\dim M_i+\sum_{j=1}^t(\dim
M_j'+\dim M_j'')).
  \end{align*}

Since $E_\sigma(C)^*$ is a self-orthogonal
code, $C_\phi$ is also self-orthogonal over the ring $\mathcal{R}$
with respect to the Hermitian inner product. This means that $M_i$
are self-orthogonal codes  of length
$c$ over $G_i$ for $i=1,\dots,m$ (with respect to the Hermitian inner product) and, for $1\leq
j\leq t$, we have $M_j''\subseteq (M_j')^\perp$ with respect to the
Euclidean inner product. This forces $\dim
M_i\leq c/2$ for $i=1,2,\dots,m$ and $\dim M_j'+\dim M''_j\leq c$.
It follows that
$$\dim E_\sigma(C)^*\leq \frac{(p-1)c}{2}+\frac{(q-1)c}{2}
+\frac{(p-1)(q-1)}{m-2+2t}((m-2)\frac{c}{2}+t
c)=\frac{c(pq-1)}{2}. $$



\section{Self-dual $[96,48,20]$ codes and permutations of order 15}

Let $C$ be a binary extremal self-dual $[96,48,20]$ code with an
automorphism $\sigma$ of order 15.  We  decompose $\sigma$ in a product of $c$
independent cycles of length 15, $t_5$ cycles of length 5, $t_3$ cycles of
length 3 and $f$ cycles of length 1. Then $\sigma^5$ and $\sigma^3$ are
automorphisms of $C$ of type $3$-$(5c+t_3,5t_5+f)$ and
$5$-$(3c+t_5,3t_3+f)$, respectively. According to (\ref{primes}),
$$3c+t_5=18, \ \ 3t_3+f=6, \ \ 5c+t_3=30 \ \mbox{or} \ 32, \ \ 5t_5+f=6
\ \mbox{or} \ 0.$$ This leads to $$t_5=0, \ \ c=6, \ \
(t_3,f)=(2,0) \ \mbox{or} \ (0,6).$$

\begin{lemma}\label{cor5} If $(t_3,f)=(2,0)$ then $C_\pi$ is the extended
 $[8,4,4]$ Hamming code. If $(t_3,f)=(0,6)$ then $C_\pi$ is the self-dual
 $[12,6,4]$ code generated by the matrix $(I_6 | I_6+J_6)$ where $I_6$ is the identity matrix and $J_6$ is the all-ones matrix of size $6$.
\end{lemma}

\begin{proof}
Let $C$ be a binary extremal self-dual $[96,48,20]$ code and let
$$\sigma=\Omega_1\Omega_2\Omega_3\Omega_4\Omega_5\Omega_6\Omega_7\Omega_8$$
be its
automorphism of order 15, where $\Omega_i=(15(i-1)+1,\dots,15i)$ for $i=1,\dots,6$,
$\Omega_7=(91,92,93)$, $\Omega_8=(94,95,96)$. Hence $C_\pi$ is a
binary self-dual code of length 8. If $x=(x_1,\dots,x_8)\in C_\pi$
then $\wt(\pi^{-1}(x))=15(x_1+\cdots+x_6)+3x_7+3x_8\equiv
3\wt(x)\pmod 4$. Since $C$ is a doubly-even code, $\wt(x)\equiv
0\pmod 4$ and $C_\pi$ must be a doubly-even code, too. The only
doubly-even self-dual code of length 8 is the extended
$[8,4,4]$ Hamming code. Its automorphism
group acts 2-transitively on the code, so we can take any pair of
coordinates for the two 3-cycles.

In the case $f=6$,  $C_\pi$ is a self-dual code of length $12$ and so its minimum weight is at most 4 by \cite{Handbook-SelfDual}. If $x=(x_1,\dots,x_{12})\in C_\pi$
then $$\wt(\pi^{-1}(x))=15(\underbrace{x_1+\cdots+x_6}_a)+\underbrace{x_7+\cdots+x_{12}}_b=15a+b\ge 20.$$ Hence $a\ge 1$ and if $a=1$ then $b=5$. It follows that $C_\pi$ is a self-dual
 $[12,6,4]$ code with a generator matrix in the form $(I_6 \ D)$. The only such code is $d_{12}^+$ (see \cite{Handbook-SelfDual}). For the structure of $d_{12}^+$ we use the terms from \cite{Huff}. This code has a defining set which means that its coordinates can be partitioned into duo's $\{l_1,l_2\}$, $\{l_3,l_4\}$, $\{l_5,l_6\}$, $\{l_7,l_8\}$, $\{l_9,l_{10}\}$, $\{l_{11},l_{12}\}$, such that its 15 codewords of weight 4 are the vectors with supports $\{l_{2i-1},l_{2i},l_{2j-1},l_{2j}\}$ where $1\le i<j\le 6$ (clusters). Since $C_\pi$ does not contain a codeword $x$ of weight 4 with $(a,b)=(1,3)$ or $(0,4)$ it turns out that $\{l_1,l_3,l_5,l_7,l_9,l_{11}\}=\{1,2,3,4,5,6\}$ and $\{l_2,l_4,l_6,l_8,l_{10},l_{12}\}=\{7,8,9,10,11,12\}$. As a basis for the code we can take the clusters $\{l_i,l_{i+1},l_{i+6},l_{i+7}\}$ for $i=1,2,\dots,5,$ with the $d$-set $\{1,7,8,9,10,11,12\}$.
Hence  $C_\pi$  has a generator matrix of shape $(I_6 | I_6+J_6)$.
\end{proof}

We consider both possibilities for the structure of $\sigma$ simultaneously. Since
$$x^{15}-1=(x-1)\underbrace{(1+x+x^2)}_{Q_3(x)}\underbrace{(1+x+x^2+x^3+x^4)}_{Q_5(x)}\underbrace{(1+x+x^4)}_{h(x)}\underbrace{(1+x^3+x^4)}_{h^*(x)},$$
we obtain
$$
    \dim E_\sigma(C)^*=2\underbrace{\dim M_1}_{\leq 3}+4\underbrace{\dim M_2}_{\leq 3}
+4(\underbrace{\dim M'+\dim M''}_{\leq 6}).$$
According to the balance principle (see \cite{JW}, \cite{pless-libro} or \cite{Handbook-SelfDual}), the dimension of the subcode of $C$ consisting of the codewords with 0's in the last six coordinates, is equal to $42=48-6$. Hence if $f=6$ then $\dim E_\sigma(C)^* = 42$. In the other case, the dimension of the subcode of $C_{\pi}\cong e_8$, consisting of the codewords with 0's in the last two coordinates, is 2 and therefore $\dim E_\sigma(C)^* = 40$.
It follows that
$$\dim M_1=2 \ \mbox{or} \ 3, \ \ \dim M_2=3 \ \mbox{and} \  \dim M'+\dim M''=6.$$
This means that $$C_\phi=M_1\oplus M_2\oplus M'\oplus M'',$$ where
$M_1$ is a Hermitian self-orthogonal $[6,2,\ge 2]$ code in the case $f=0$ and a self-dual $[6,3,\ge 2]$ code in the case $f=6$ over the field
$G_1\cong \F_4$, $M_2$ is a Hermitian self-dual $[6,3,d_2]$ code
over $G_2\cong\F_{16}$, $M'$ is a linear $[6,k',d']$ code over
$H\cong\F_{16}$ and $M''=(M')^\perp$ is its dual with respect to
the Euclidean inner product. If $v$ is a codeword of weight $t$ in $M_2$, $M'$ or $M''$ then the vectors $\phi^{-1}(v)$, $\phi^{-1}(xv)$, $\phi^{-1}(x^2v)$ and $\phi^{-1}(x^3v)$ generate a binary code of dimension 4 and effective length $15t$. It is a subcode of $C$ and therefore its minimum distance should be at least 20. Since binary codes of length 30, dimension 4 and minimum distance $\ge 20$ do not exist \cite{table-Grassl},
$d_2=3$ or 4, $d'\ge 3$ and the minimum distance of $M''$ is at
least 3. In the following we list the three possible cases for $M'$ and $M''$
where $$ e = e(x)=x^{12}+x^9+x^8+x^6+x^4+x^3+x^2+x$$
is the identity of the field
 $H=\{0, e, xe, x^2e,\ldots, x^{14}e\}$.

\begin{enumerate}
\item $M'$ is an MDS $[6,2,5]$ code and $M''$ is its dual MDS
$[6,4,3]$ code. It is well known that any MDS $[n,k,n-k+1]$ code
over $\F_q$ is an $n$-arc in the projective geometry $PG(k-1,q)$.
There are exactly four inequivalent $[6,2,5]$ MDS codes over
$\F_{16}$ \cite{tableMDS} (their dual codes correspond to the
6-arcs in $PG(3,16)$). We list here generator matrices of these
codes:
\[\left(\begin{array}{cccccc}
e&0&e&e&e&e\\
0&e&e&xe&x^2e&x^3e\\
\end{array}\right) \ \ \ \
\left(\begin{array}{cccccc}
e&0&e&e&e&e\\
0&e&e&xe&x^2e&x^4e\\
\end{array}\right)\]
\[\left(\begin{array}{cccccc}
e&0&e&e&e&e\\
0&e&e&xe&x^3e&x^7e\\
\end{array}\right) \ \ \ \
\left(\begin{array}{cccccc}
e&0&e&e&e&e\\
0&e&e&xe&x^3e&x^{11}e\\
\end{array}\right)\]

\item $M'$ and $M''$ are both MDS $[6,3,4]$ codes. According to \cite{tableMDS}, there are 22 MDS codes with the needed parameters over $\F_{16}$ (they correspond to the 6-arcs in
$PG(2,16)$). We consider generator matrices of these codes in the form
\[\left(\begin{array}{cccccc}e&0&0&e&e&e\\
0&e&0&e&x^{a_1}e&x^{a_2}e\\
0&0&e&e&x^{a_3}e&x^{a_4}e\\
\end{array}\right), \ \ a_i\in\{1,2,\dots,14\}, \ i=1,2,3,4. \]
 Note that $a_i\ge 1$ for $i=1,2,3,4$ since the minimum distance of $M'$ is 4.
We calculated the weight distributions and the automorphism groups of
$\phi^{-1}(M'\oplus M'')$ for all 22 codes $M'$. The results are listed in Table \ref{TableMDS634}. Five of the binary codes have minimum distance 24,
and six of them have minimum distance 20.

\begin{table}[!ht]
\centering
\caption{The $[90,24]$ codes in case 2}\label{TableMDS634} {\footnotesize
\begin{tabular}{c@{\ }|@{\ }c@{\ }|@{\ }c@{\ }|@{\ }c@{\ }|@{\ }c@{\ }|@{\ }c@{\ }|@{\ }c@{\ }|@{\ }c@{\ }|@{\ }c@{\ }|@{\ }c@{\ }|@{\ }c@{\ }|@{\ }c}
$(a_1,a_2,a_3,a_4)$&$A_{16}$&$A_{20}$&$A_{24}$&$A_{28}$&$A_{32}$&$A_{36}$&$|\Aut|$\\
\hline
$(1,2,2,1)$&270&0&5400&15840&195345&941400&1440\\
$(1,2,2,4)$&60&120&2730&18480&189885&950280&240\\
$(1,2,2,5)$&15&30&2070&17535&187815&963480&30\\
$(1,2,2,6)$&45&180&1935&17505&183015&975420&90\\
$(1,2,2,8)$&45&0&2580&15660&188715&965040&240\\
$(1,2,2,9)$&15&30&2130&17355&187575&965160&30\\
$(1,2,3,1)$&30&120&2430&19650&192105&937200&120\\
$(1,2,3,6)$&-&-&2325&16320&192585&953040&60\\
$(1,2,3,7)$&-&60&1875&17955&189465&956220&30\\
$(1,2,3,8)$&-&-&2145&17340&190185&956400&30\\
$(1,2,3,12)$&-&60&1965&18060&187545&960120&60\\
$(1,2,4,6)$&-&60&2040&17910&187485&959400&60\\
$(1,2,5,7)$&-&90&1830&18390&186405&963900&30\\
$(1,2,6,1)$&60&0&3090&17400&194205&941400&240\\
$(1,2,9,1)$&30&120&2910&17250&196425&933840&120\\
$(1,2,12,1)$&90&360&3240&23940&192825&909720&720\\
$(1,3,2,6)$&-&-&2325&16320&192585&953040&60\\
$(1,3,3,2)$&-&180&1665&18720&185625&960840&90\\
$(1,3,7,2)$&-&-&2295&16830&191745&950040&180\\
$(1,3,7,10)$&-&180&1755&18450&185265&963360&360\\
$(1,3,11,8)$&-&-&2730&14100&197925&944760&600\\
$(5,10,10,5)$&450&0&14580&16200&329625&507960&259200\\
\end{tabular}}
\end{table}

\item $M'$ and $M''$ are both $[6,3,3]$ codes.
We consider generator matrices of $M'$ in the form
\[\left(\begin{array}{cccccc}e&0&0&0&e&e\\
0&e&0&e&\beta_1&\beta_2\\
0&0&e&e&\beta_3&\beta_4\\
\end{array}\right), \ \ \beta_i\in H, \ i=1,2,3,4,\]
where $\beta_i=x^{b_i}e$, $b_i\in\{0,1,\dots,14\}$,  or $\beta_i=0$, $i=1,2,3,4$.


We calculated that there are 18 inequivalent $[6,3,3]$ codes $M'$ over $\F_{16}$ such that $d(\phi^{-1}(M'\oplus M''))\geq 20$.
The weight distributions and the automorphism groups of $\phi^{-1}(M'\oplus M'')$ for all 18 codes are listed in Table \ref{TableMDS633}. Ten of the binary codes have minimum distance 24, and eight of them have minimum distance 20.

\begin{table}[!ht]
\centering
\caption{The $[90,24]$ codes in case 3}\label{TableMDS633} {\footnotesize
\begin{tabular}{c@{\ }|@{\ }c@{\ }|@{\ }c@{\ }|@{\ }c@{\ }|@{\ }c@{\ }|@{\ }c@{\ }|@{\ }c@{\ }|@{\ }c@{\ }|@{\ }c@{\ }|@{\ }c@{\ }|@{\ }c@{\ }|@{\ }c}
$(b_1,b_2,b_3,b_4)$&$A_{20}$&$A_{24}$&$A_{28}$&$A_{32}$&$A_{36}$&$|\Aut|$\\
\hline
$(0,0,0,7)$&-&2250&17640&187605&960120&180\\
$(0,0,2,3)$&-&2070&18060&187125&963960&60\\
$(0,0,2,6)$&-&1950&18420&187605&960600&30\\
$(0,2,2,9)$&-&2175&17670&188625&957480&60\\
$(0,2,3,4)$&-&2070&17730&189285&958080&15\\
$(0,2,3,7)$&-&2025&17865&189465&956820&15\\
$(0,2,3,11)$&-&2070&18030&187485&962280&30\\
$(0,2,3,12)$&-&2010&18210&187725&960600&15\\
$(0,2,4,7)$&-&2190&16890&191925&953040&30\\
$(0,2,4,13)$&-&2100&17640&189165&958920&60\\
$(0,0,0,2)$&90&1755&18900&184545&968940&90\\
$(0,0,2,5)$&30&1905&18630&186585&961260&30\\
$(0,0,2,9)$&30&2025&18270&186105&964620&30\\
$(0,0,3,5)$&30&1935&18540&186465&962100&30\\
$(0,2,2,3)$&60&2055&18030&186225&963600&30\\
$(0,2,2,5)$&60&1935&18015&188265&958860&30\\
$(0,2,2,8)$&180&1800&18630&184365&962760&180\\
$(0,2,4,0)$&90&1830&18630&185445&964860&30\\
\end{tabular}}
\end{table}

\end{enumerate}


In the following $G_1$ is the field with four elements and identity
$$ e_1=x + x^2 + x^4 + x^5 + x^7 + x^8 + x^{10} + x^{11} + x^{13} + x^{14},$$
and $G_2$ the field with $16$ elements and identity
$$ e_2=x + x^2 + x^3 + x^4 + x^6 +
x^7 + x^8 + x^9 + x^{11} + x^{12} + x^{13} + x^{14},$$ defined in the beginning of this section.
Furthermore $\mu_2=x^{11} + x^{10} + x^6 + x^5 + x + 1$ is a generator of $G_2$.

According to \cite{Yorgov-Yankov}, there are two Hermitian
self-dual $[6,3,d\ge 3]$ codes over $\F_{16}$ up to the
equivalence defined in the following way: Two codes are equivalent
if the second one is obtained from the first one via a sequence of
the following transformations:
\begin{itemize}
\item a substitution $x\to x^t$, $t=2,4,8$;
\item a multiplication of any coordinate by $x$;
\item a permutation of the coordinates.
\end{itemize}

Their generator matrices are
\[
H_1=\left(\begin{array}{cccccc}
e_2 & 0 & 0 & 0 & \mu_2^5 & \mu_2^{10} \\
0 & e_2 & 0 & \mu_2^5 & \mu_2^5 & e \\
0 & 0 & e_2 & \mu_2^{10} & e_2 & \mu_2^{10} \\
\end{array}\right), \ \ \ \ \
H_2=\left(\begin{array}{cccccc}
e_2 & 0 & 0 & e_2 &\mu_2^5 & \mu_2^5 \\
0 & e_2 & 0 & e_2 & \mu_2^2 & \mu_2^8 \\
0 & 0 & e_2 & e_2 & \mu_2^6 & \mu_2^9 \\
\end{array}\right).\] 

We fix the $M'\oplus M''$ part of the generator matrix and
consider all possible generator matrices for the $M_2$ part. Note
that even if the matrices generate equivalent
codes $M_2$ the codes generated by $M'\oplus M''\oplus M_2$ may not be
equivalent. We consider the two possible matrices for the $M_2$
part under the products of the following maps: 1) a permutation $\tau\in S_6$ of
the 15-cycle coordinates; 2) multiplication of each of the 6
columns by nonzero element of $F_{16}$; 3) automorphism of the
field ($x\to x^t$, $t=2,4,8$). After computing all possible
generator matrices we obtain exactly 675 inequivalent $[90, 36,
20]$ binary codes: 232 from the first matrix $H_1$, and 443 from the
second $H_2$. These codes have automorphism groups of orders 15
(557 codes), 30 (111 codes), 45 (2 codes) and 90 (5 codes).

 Next we separate the cases $f=0$ and $f=6$.
\begin{itemize}
\item[Case] $f=0$: Let first add the fixed subcode. According to Lemma \ref{cor5},
the code $\pi(F_\sigma(C))$ is equivalent to the extended Hamming
$[8,4,4]$ code $H_8$. As we already mentioned in the proof of Lemma \ref{cor5}, we can take any pair of
coordinates for the 3-cycles. Then we consider all $6!=720$
permutation of the 15-cycles that can lead to different subcodes.
Only 47 of the constructed codes $\phi^{-1}(M'\oplus M''\oplus M_2)\oplus F_\sigma(C)$ have minimum distance $d'=20$ (we list
the number of their codewords of weights 20 and 24 and the order of
the automorphism groups in Table \ref{n90_40_20}).

\begin{table}[!ht]
\centering
\caption{The $[96,40,20]$ codes}\label{n90_40_20}
{\footnotesize
\begin{tabular}{|c@{\ }|@{\ }c@{\ }|@{\ }c@{\ }|@{\ }c@{\ }||c@{\ }|@{\ }c@{\ }|@{\ }c@{\ }|@{\ }c@{\ }||c@{\ }|@{\ }c@{\ }|@{\ }c@{\ }|@{\ }c@{\ }|}
&$A_{20}$&$A_{24}$&$|\Aut|$&&$A_{20}$&$A_{24}$&$|\Aut|$&&$A_{20}$&$A_{24}$&$|\Aut|$\\
\hline
$C_{96,40,1}$&48735&4206590&1620&$C_{96,40,17}$&47925&4216010&540&$C_{96,40,33}$&48045&4213610&540\\
\hline
$C_{96,40,2}$&49545&4197410&1620&$C_{96,40,18}$&48105&4213730&540&$C_{96,40,34}$&48420&4209320&540\\
\hline
$C_{96,40,3}$&47835&4217030&1620&$C_{96,40,19}$&48600&4207760&540&$C_{96,40,35}$&47760&4216160&540\\
\hline
$C_{96,40,4}$&47940&4214600&540&$C_{96,40,20}$&48420&4208120&540&$C_{96,40,36}$&48780&4204760&540\\
\hline
$C_{96,40,5}$&48405&4209530&540&$C_{96,40,21}$&47325&4220810&540&$C_{96,40,37}$&48510&4209500&540\\
\hline
$C_{96,40,6}$&47805&4214810&540&$C_{96,40,22}$&47595&4216070&540&$C_{96,40,38}$&47460&4217720&540\\
\hline
$C_{96,40,7}$&47205&4222490&540&$C_{96,40,23}$&48345&4209650&540&$C_{96,40,39}$&48330&4210100&1080\\
\hline
$C_{96,40,8}$&48690&4204820&540&$C_{96,40,24}$&47925&4213370&540&$C_{96,40,40}$&47415&4221950&1080\\
\hline
$C_{96,40,9}$&47265&4220450&540&$C_{96,40,25}$&47835&4215110&540&$C_{96,40,41}$&48315&4210550&540\\
\hline
$C_{96,40,10}$&47580&4216520&540&$C_{96,40,26}$&47790&4214780&540&$C_{96,40,42}$&47490&4218740&540\\
\hline
$C_{96,40,11}$&47565&4219370&1080&$C_{96,40,27}$&49410&4200020&540&$C_{96,40,43}$&49140&4201880&540\\
\hline
$C_{96,40,12}$&48255&4212110&540&$C_{96,40,28}$&48225&4210610&540&$C_{96,40,44}$&48330&4212500&1080\\
\hline
$C_{96,40,13}$&48555&4207190&540&$C_{96,40,29}$&48360&4209920&540&$C_{96,40,45}$&48870&4212860&1080\\
\hline
$C_{96,40,14}$&48165&4211690&1080&$C_{96,40,30}$&48600&4214000&1080&$C_{96,40,46}$&47970&4213220&540\\
\hline
$C_{96,40,15}$&48555&4206710&1080&$C_{96,40,31}$&47775&4215230&540&$C_{96,40,47}$&47925&4215050&1080\\
\hline
$C_{96,40,16}$&48630&4205900&540&$C_{96,40,32}$&49815&4194350&1620&&&&\\
\hline
\end{tabular}}
\end{table}

Next we add the $M_1$ part, that is a Hermitian self-orthogonal $[6, 2, \geq 2]$ code over the field $G_1\cong\F_4$. One can easily compute all such codes up to equivalence. There are exactly 4 inequivalent such codes
with generator matrices
$$H_3=\left(\begin{array}{cccccc}
e_1&0&e_1&0&0&0\\
0&e_1&0&e_1&0&0\\
\end{array}\right), \ \ \
H_4=\left(\begin{array}{cccccc}
e_1&0&e_1&e_1&e_1&0\\
0&e_1&e_1&xe_1&x^2e_1&0\\
\end{array}\right),$$
$$H_5=\left(\begin{array}{cccccc}
e_1&0&e_1&0&0&0\\
0&e_1&0&e_1&e_1&e_1\\
\end{array}\right), \ \ \
H_6=\left(\begin{array}{cccccc}
e_1&0&0&e_1&e_1&e_1\\
0&e_1&e_1&0&e_1&e_1\\
\end{array}\right).$$

We fix the generator matrices of the 47 codes and consider the
matrices $H_3,H_4,H_5,H_6$ under compositions of the following
transformations: 1) a permutation $\tau\in S_6$ of the 15-cycle coordinates;
2) multiplication of each of the 6 columns by a nonzero element of
$G_1$; 3) automorphism of the field ($x\to x^2$). Thus we
construct binary $[96,44]$ codes. Our computations show that none
of these codes has minimum distance $d\ge 20$.

\item[Case] $f=6$:
 Now we add the $M_1$ part, which is a Hermitian quaternary self-dual code of length 6 over the field $G_1\cong\F_4$. There are two inequivalent codes of this length - $i_2^3$ with minimum weight 2 and $h_6$ with minimum weight 4 (see \cite{Handbook-SelfDual}). All 675 inequivalent $[90, 36,20]$ codes combined with the binary images of the different copies to both quaternary self-dual codes give binary self-orthogonal $[90,42,\le 16]$ codes.
\end{itemize}

This proves
Theorem \ref{Main_theorem} which states that
a binary doubly-even $[96, 48, 20]$ self-dual code with an automorphism of order 15 does not exist. The
 calculations were done with the \textsc{GAP Version 4}
 software system \cite{GAP4} and the program \textsc{Q-Extension} \cite{Serdica_q_ext}.

\section{On the automorphism of type $3$-$(28,12)$}

In this section we fill
 a gap in the literature caused by a missing proof on the nonexistence of an extremal self-dual code of length 96 having an automorphism of type $3$-$(28,12)$.
In  \cite{JW}, the authors used this assertion in their proof of the main theorem.

\begin{proposition}
A binary doubly-even $[96,48,20]$
self-dual code with an automorphism of type $3$-$(28,12)$ does not exist.
\end{proposition}

\begin{proof}
Suppose that $C$ is a self-dual $[96,48,20]$ code and $\sigma$ is an automorphism of $C$ of type 3-(28,12).
 Then
$C_\pi$ is a self-dual $[40,20,8]$ code. Without loss of generality, we can take the last 12 coordinates for the fixed points. So $C_{\pi}$ has a generator matrix of the form
\begin{equation}\label{matGpi36}
G_{\pi} =\left(\begin{array}{cc}
A & O\\
D & I_{12}
\end{array}\right),\end{equation}
where $A$ is an $8\times 28$ matrix which generates a doubly-even $[28,8,\ge 8]$ code ${\cal A}$ with  dual distance $d_{\cal A}^\perp\ge 3$.
 Using the MacWilliams equalities we see that the possible weight distribution for this code is $$W_{\cal A}(y)=1+\lambda y^8+(142-3\lambda-\mu)y^{12}+(95+3\lambda+3\mu)y^{16}+(18-\lambda-3\mu)y^{20}+\mu y^{24},$$
and the number of codewords of weight 3 in its dual code is $\nu=2\lambda-2\mu-4$.

 Let us consider the partitioned weight enumerator $A_{ij}$ for the code $C_{\pi}$, where $0\le i\le 28$ and $0\le j\le 12$. We use the following restrictions:
 \begin{itemize}
 \item If $3i+j\not\equiv 0\pmod 4$ then $A_{ij}=0$. 
 \item If $0<i+j<8$ or $32<i+j<40$ then $A_{ij}=0$.
 \item If $0<3i+j<20$ or $76<3i+j<96$ then $A_{ij}=0$.
 \item $A_{i0}=\alpha_i$, where $\{\alpha_i, i=0,\dots,28\}$ is the weight distribution of ${\cal A}$. 
 \item $A_{ij}=A_{28-i,12-j}$, $i=0,\dots,28$, $j=0,\dots,12$.
  \end{itemize}

  According to the MacWilliams identities for coordinate partitions (see \cite{Sim}) and the above restrictions, we obtain the following system of linear equations
  $$
  2^{20}A_{s,0}=\sum_{i=0}^{28}\sum_{j=0}^{12}\mathcal{K}_s(i;28)\mathcal{K}_0(j;12)A_{i,j}; \ \ \ \ \ 2^{20}A_{s,1}=\sum_{i=0}^{28}\sum_{j=0}^{12}\mathcal{K}_s(i;28)\mathcal{K}_1(j;12)A_{i,j}$$
  $$\iff 2^{20}A_{s,0}=\sum_{i=0}^{28}\sum_{j=0}^{12}\mathcal{K}_s(i;28)A_{i,j}; \ \ \ \ \  2^{20}A_{s,1}=\sum_{i=0}^{28}\sum_{j=0}^{12}\mathcal{K}_s(i;28)(12-2j)A_{i,j}$$
   $$\iff 2^{20}A_{s,0}=\sum_{i=0}^{28}\sum_{j=0}^{12}\mathcal{K}_s(i;28)A_{i,j};  \ \ \ \ \ 2^{19}(12A_{s,0}-A_{s,1})=\sum_{i=0}^{28}\sum_{j=0}^{12}j\mathcal{K}_s(i;28)A_{i,j}$$

  Solving this system with respect to 25 of the unknowns by using Computer Algebra System Maple, we obtain $\lambda=-1$, a contradiction.
\end{proof}




\end{document}